\documentclass[12pt]{article}

\usepackage{amsfonts,amsmath}
\usepackage{graphics}
\usepackage{epsfig}
\usepackage{times}
\usepackage{authblk}
\newtheorem{theorem}{Theorem}
\newtheorem{lemma}{Lemma}

\newtheorem{definition}{Definition}
\newtheorem{corol}{Corollary}

\newtheorem{remark}{Remark}
\newtheorem{proof}{Proof}

\def\x{{x}}
\def\y{y}
\def\X{{X}}

\def\C{{\cal C}}

\def\opt{\ell}
\def\tk{s}
\newcommand{\R}[1]{{\cal R}_{#1}}
\newcommand{\T}[2]{{\cal T}_{#1,#2}}
\renewcommand{\l}[1]{c_{#1}}
\renewcommand{\ll}[2]{c_{#1}(#2)}
\newcommand{\lll}[2]{c'_{#1}(#2)}

\newcommand{\remove}[1]{}

\title{Minimum Weight Dynamo and Fast  Opinion Spreading}

\author{Sara Brunetti, Gennaro Cordasco, Luisa Gargano, Elena Lodi, Walter Quattrociocchi}

\begin{document}

\maketitle
\begin{abstract}We consider the following multi--level opinion spreading model on networks.  Initially, each node gets a weight from the set $\{0,\ldots,k-1\}$, where such a weight stands for the individuals conviction of a new idea or product. 
Then, by proceeding to rounds, each node updates its weight according to the weights of its neighbors.
We are interested in  the initial assignments of weights leading each node to get the  value $k-1$
--e.g. unanimous maximum level acceptance-- within
a given number of  rounds.
We determine lower bounds on  the sum of the initial weights of  the nodes  under the  irreversible  simple majority rules,
where a node increases its weight if and only if the majority of its neighbors have a weight that is higher than its own one.  Moreover,  we provide constructive tight upper bounds for some class of regular topologies: rings, tori, and cliques.

\end{abstract}
\vspace*{-1truecm}
\section{Introduction} \vspace*{-0.2truecm}
New opinions and behaviors  usually spread gradually
through social networks. In 1966 a classical study showed how doctors'
willingness to prescribe a new antibiotic  diffused through professional contacts. A similar
pattern can be detected in
a variety of innovations: Initially a few innovators adopt, then people in contact with the
innovators get interested and then  adopt,  and so forth until eventually the
innovation spreads throughout the society.
A classical question is then  how many
innovators are needed, and how they need to be disposed, in order to get a fast unanimous adoption \cite{V95}.

In the wide set of the  information spreading models,  the first computational study about information diffusion \cite{Granovetter85} used the {\em linear threshold model} where  the threshold triggering the adoption  of a new idea to a node is given by the majority of its active neighbors.

Recently,  {information spreading} has been intensively studied also in the context of {\em viral marketing},
 which  uses  social networks to   achieve  marketing objectives  through self-replicating viral processes, analogous to the spread of viruses. The goal here is  to create a marketing message that can initially  convince a selected set of people and then  spread to  the whole network in a short period of time  \cite{Domingos2001}.
 One problem in  viral marketing is  the  {\em target set selection problem } which
  asks for identifying the minimal number of nodes  which can  activate, under some conditions, the whole network  \cite{EK10}.
The target set selection problem has been proved to be NP-hard through a reduction to the node cover problem \cite{Kempe03}.
Recently, inapproximability  results of  opinion spreading problems have been presented in \cite{Chen09}.

In this paper, we consider the following novel opinion spreading model.  Initially, each node is assigned a weight from the set $\{0,\ldots,k-1\}$; where the weight of a node represents the level of acceptance of the opinion
by the actor represented by the node itself.
Then, the process proceeds in synchronous rounds where each node updates its weight depending on the weights of its neighbors.
We are interested in  the initial assignments of weights leading to the all--($k-1$)  configuration  within
a given number of  rounds. The goal is to minimize the sum of the initial weights of the nodes.

Essentially, we want everyone to completely accept the new opinion within a given time bound while minimizing the initial convincing effort -- i.e, the sum  of the initial  node weights.

We are interested in the case in which the  spreading is
essentially a one-way process: once an agent has adopted an opinion (or behavior, innovation, $\ldots$), she sticks with it.
These are usually referred as {\em irreversible} spreading processes.
\vspace*{-0.2truecm}
\paragraph{\textbf{Dynamic Monopolies and Opinion Spreading.}}
In a different scenario,  spreading processes have been studied under the name of dynamic monopolies.
Monopolies were initially introduced to deal with faulty nodes  in distributed computing systems.
A monopoly in a graph is  a subset $M$ of nodes such that each other node of the graph has a prescribed number of neighbors belonging to $M$.
 The problem of finding monopolies in graphs has been widely studied,  see for example  \cite{BermondBPP96}, \cite{Linia93}, and \cite{Mishra2003126} for connections with minimum dominating set problem.

Dynamic monopolies or shortly {\em dynamo} were introduced by Peleg \cite{Peleg98}.
A strong relationship between opinion spreading problems, such as the target set selection, and  dynamic monopolies exists.
Indeed, they can be used to model the irreversible spread of opinions in social networks. \\
Dynamic monopolies have been  intensively studied with respect to the bounds of the size of the monopolies, the time needed to converge into a fixed point, and topologies over which the interaction takes place  \cite{BermondBPP03}, \cite{Bermond98}, \cite{Santoro03}, \cite{Kulich2011Dynamic}, \cite{NayakPS92}, \cite{Peleg02}.
\vspace*{-0.2truecm}
\paragraph{\textbf{Our results: Weighted opinion spreading.}}
We model the  opinion spreading process considered in this paper by means of weighted dynamos.

We  extend the setting of dynamos from 2 possible weights (denoting whether a node has accepted the opinion or  not) to $k$ levels  of opinion acceptance (a different extension has been studied in \cite{BrunettiLQ11}).
Initially, each node has a weight (which represents the node initial level of acceptance of the opinion) in the set $\{0,\ldots, k-1\}$. Then, each  node updates its weight by increasing  it of one unit if the weights of the simple majority of its neighbors is larger than its own.   We call $k$-dynamos, the initial weight assignments which lead each node in the network to have maximum weight $k-1$.   We are interested in   the minimum weight (i.e. the sum of the weight initially assigned to the nodes) of a $k$-dynamo. We focus on both the weight and the time (e.g., number of rounds needed to reach the final configuration); namely,   we study $k$-dynamos of minimum weight which converge into at most $t$ rounds.\vspace*{-0.2truecm}
\paragraph{\textbf{Paper organization.}}
In Section \ref{not}, we formalize the model and fix the notation. In  Section \ref{sec3}, we
determine lower bounds on the weight of $k$-dynamos  which converge into at most $t$ rounds.
Section \ref{sec4}  provides tight constructive upper bounds for rings, tori  and cliques.
In the last section, we conclude and state a few open problems.

\section{The Model}\label{not}
Let $G=(V,E)$ be an undirected connected graph.
For each $v\in V$, we denote by  $N(v)=\{u\in V \ | \  \{u,v\}\in E\}$ the neighborhood of $v$ and by $d(v)=|N(v)|$ its cardinality (i.e., the degree of $v$).
\\
We assume the nodes of $G$ to be weighted by the set  $[k]=\{0,1,\ldots,k-1\}$ of  the first $k\geq 2$ integers.
For each $v\in V$ we denote by $\l{v}\in [k]$ the weight assigned to a given node $v$.
\begin{definition}
A \textit{configuration} $\C$ on $G$ is a partition of $V$ into $k$ sets $\{V_0,V_1,\ldots,V_{k-1}\}$, where
$V_j=\{v\in V \ | \ \l{v}=j \}$ is the set of nodes   of weight $j$.
The weight $w(\C)$ of  $\C$ is the weighted sum of its nodes
 $$w(\C)=\sum_{j=0}^{k-1} j\times |V_j|=\sum_{v\in V} \l{v}.$$
\end{definition}
Consider the following node weighting game played on $G$ using the set of weights $[k]$ and   a threshold value $\lambda$
(for some   $0<\lambda \leq 1$):
\begin{itemize}
\item[]
 In the initial configuration, each node has  a  weight in $[k]$.
 Then node weights are updated in synchronous  rounds (i.e., round $i$ depends on round $i-1$ only). Let  $ \ll{v}{i}$ denote the weight of node $v$ at the end of round $i\geq 0$;
 during round $i\geq 1$, each node  updates its weight according to the weight of its neighbors at  round $i-1$. Specifically,    each node $v$
\begin{itemize}
\item first computes the number $n^+(v)=|\{u\in N(v) \ | \ \ll{u}{i-1}>\ll{v}{i-1}\}|$  of neighbors  having a weight larger  than its current one $\ll{v}{i-1}$;
\item then, it applies the following {\em irreversible rule: }
\hspace*{-0.3truecm}
$$  \ll{v}{i} = \left\{  \begin{array}{l l}
    \ll{v}{i{-}1}+1 &  \text{if  $n^+(v)\geq \left\lceil \lambda d(v) \right \rceil$} \\
    \ll{v}{i{-}1} &  \text{otherwise}
  \end{array} \right.$$
\end{itemize}	
\item[] We denote the initial configuration by $\C^0$ and the configuration at round $i$ by  $\C^i$.
\end{itemize}
We are interested into initial configurations that converge to the unanimous all-($k-1$)s configuration -- i.e., there exists a round $t^*$ such that for each $i\geq t^*$ and for each node $v,$  it holds $\ll{v}{i}=k-1$.
Such configurations are named $k$-weights dynamic monopoly (henceforth  \textit{$k$-dynamo}).

A \textit{($k,t$)-dynamo} is a $k$-dynamo which reaches its final configuration within $t$ rounds, that is, $\ll{v}{i}=k-1$ for each node $v \in V$ and $i \geq t.$ An example of ($k,t$)-dynamo, with $\lambda=1/2$,  is depicted in Figure \ref{example}.
Given a graph $G$, a set of weights $[k]$, a threshold $\lambda$, and an integer $t>0$, we aim for a minimum weight ($k,t$)-dynamo.

\begin{definition}
A ($k,t$)-dynamo on a graph $G$ with threshold $\lambda$ is \textit{optimal} if its  weight is minimal among all the ($k,t$)-dynamos  for the graph $G$ with threshold $\lambda$.
\end{definition}

\begin{figure}
	\centering
\includegraphics[width=0.70\linewidth]{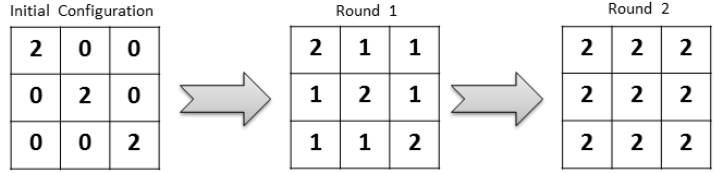}
\vspace*{-0.4truecm}
		\caption{A ($3,2$)-dynamo on a $3 \times 3$ Tori ($\lambda=1/2$): Starting from the initial configuration (left), two rounds are needed to reach the final all-($2$)s configuration.\vspace*{-0.5truecm}			\label{example}}
\end{figure}

\section{Time bounded dynamos}\label{sec3} \vspace*{-0.2truecm}
In this section we provide a lower bound on the weight of a ($k,t$)--dynamo and
study the minimum value of $t$ for which an optimal ($k,t$)--dynamo coincides with a $k$--dynamo.
\subsection{Preliminary Results} \vspace*{-0.2truecm}

\begin{definition} \label{def3}
Consider an undirected connected graph  $G=(V,E)$. Let $k\geq 2$ and $t \geq 1$ be integers and $0<\lambda\leq1$. An initial configuration $\C$ for $G$ is called ($k,t$)-simple-monotone if $V$ can be partitioned into $t+1$ sets $\X_{-s},\X_{-s+1},\ldots\X_{k-1}$ (here $\tk=t-k+1$) where $\X_{k-1}\neq \emptyset$ , and for each $v \in \X_i$

 (i) $\ $ $\ll{v}{0}=\max(i,0)$;

 (ii) $v$ has at least $\lceil \lambda d(v) \rceil$ neighbours in $\bigcup_{j=i+1}^{k-1}\X_j.$
\end{definition}
\begin{lemma}\label{lem:ub}
Any ($k,t$)-simple-monotone configuration for an undirected connected graph $G$ is a ($k,t$)-dynamo  for  $G$.
\end{lemma}
\begin{proof}{}
We show  that for each  $i=-\tk,-\tk+1,\ldots,k-1$ (here $\tk=t-k+1$) and  $j=0,\ldots,t$ and for each $u\in X_i$ 

$$ \ll{u}{j}= \left\{  \begin{array}{l l}
    \min(j+i,k-1) &  \text{if  $j+i>0$} \\
    0 &  \text{otherwise.}
  \end{array} \right.	$$
We prove this statement by induction on $i$ from $k-1$ back to $-\tk$.
For $i=k-1$ the nodes in $X_{k-1}$ have weight $k-1$ from the initial configuration and the statement is trivially true for each round $j$. \\
Assume now that the statement is true for any $r > i$. For each  $u \in X_i$, we know that $u$ has at least  $\lceil \lambda d(v) \rceil$ neighbours which belong to $\bigcup_{r=i+1}^{k-1} \X_r$. By induction, each of this neighbor nodes, for each round $j$ has a weight greater or equal to $\min(j+i+1,k-1)$ if $j+(i+1)>0$.

Hence, $u$ preserves its weight $c_u(j)=max(i,0)=0$
until it increases its weight at each round $j$ such that $j+(i+1)>1$
(i.e. $j+i>0$) and $c_u(j) < k-1$; as a result each node in $\X_i$ has weight  $\min(j+i,k-1)$ whenever $j+i>0$, for each  $j=0,1,\ldots,t$.
\\
The Lemma follows since at round $t$, $i+j=i+t\geq -\tk+t=k-1>0$. Hence, all the nodes will have weight $\min(i+t,k-1)=k-1.$.
\end{proof}

\begin{lemma}\label{lem:lem1}
Let $G=(V,E)$ be an undirected connected graph. There exists an  optimal ($k,t$)-dynamo for $G$ which is a ($k,t$)-simple-monotone configuration for $G$.
\end{lemma}
\begin{proof}{}
Let $\C$ be an optimal  ($k,t$)-dynamo.
Define a new configuration $\C'$  as follows: Let $\tk=t-k+1$, for $i=k-1,k-2,\ldots,-\tk$, let $\X_i$ be the set of nodes that, starting with configuration $\C$, reaches permanently the weight $k-1$ at round $k-1-i$,
that is,

$\X_i=\{u\in V\ |\ \ll{u}{k-2-i}\neq k-1, \hbox{ and } \ll{u}{j}=k-1 \hbox{ for each } j\geq k-1-i\}.$
\\
In $\C'$, for each $u\in X_i$ set $ \lll{u}{0}= \max(i,0).$\\
Notice that since $\C$ is a $k$-dynamo which converges into $t$ rounds,  $\{\X_{-\tk},\X_{-\tk+1},\ldots,\X_{k-1}\}$ is a partition of $V$ and $\X_{k-1}\neq \emptyset$.
We now show  that $w(\C')\leq w(\C)$ and $\C'$ is a ($k,t$)-simple-monotone configuration for $G$.
Clearly, \vspace*{-0.2truecm}
\begin{itemize}
\item [(a)]for each index $i\leq 0,$   and for each  $u\in \X_i,$  $\ll{u}{0}\geq \lll{u}{0}=0$;
\item [(b)]for each  $i>0$  and for each  $u\in \X_i$ we have $\ll{u}{0}\geq \lll{u}{0}=i$ (otherwise $u$ cannot reach the final weight $k-1$ by round $k-1-i$, since the weight of a node increases by at most $1$ at each round).  \vspace*{-0.2truecm}
\end{itemize}
By using (a) and (b) above we have that $w(\C')\leq w(\C)$. It remains to show that $\C'$ is a ($k,t$)-simple-monotone configuration for $G$.
By construction, $\C'$ satisfies point (i) of Definition \ref{def3}.
Moreover, for each  $u \in X_i$, we know that $u$ in the configuration $\C$ reaches the weight $k-1$ at round $k-1-i$. Hence at least $\lceil \lambda d(v)\rceil$ of its neighbors have weight $k-1$ at round $k-1-i-1=k-1-(i+1)$, that is at least $\lceil \lambda d(v)\rceil$ of its neighbors belong to $\bigcup_{j=i+1}^{k-1} \X_j$. Hence, point  (ii) of Definition \ref{def3} also holds. 
\end{proof}


\vspace*{-0.4truecm}

\subsection{A Lower Bound} \vspace*{-0.2truecm}
\begin{theorem} \label{th1}
Consider  an  undirected connected graph $G=(V,E)$ and let $k\geq 2$ and $t \geq 1$ be integers. Any ($k,t$)-dynamo $\C$, with $\lambda=1/2$, has weight
$$ w(\C)\geq  \left\{  \begin{array}{l  l}
     \frac{|V|}{2 \rho (\opt+\tk+1) +1}\times \left( k{-}1+\rho\opt(\opt{+}1)  \right) \\
      \ \ \ \ \text{ where } \opt=\left\lfloor\frac{\sqrt{(2\rho\tk+ \rho +1)^2+4\rho(k-1)}-(2\rho\tk+ \rho +1)}{2\rho}\right\rfloor & \text{if  $t\geq k-1$} \\
     \frac{|V|}{2 \rho (\opt+\tk+1) +1}\times \left( k{-}1+\rho(\opt(\opt{+}1)-\tk(\tk{+}1))  \right) \\
     \ \ \ \ \text{ where } \opt=\left\lfloor\frac{\sqrt{4\rho (t+1)+(\rho-1)^2}-(2\rho\tk+ \rho +1)}{2\rho}\right\rfloor &\text{otherwise,}
  \end{array} \right.$$
 where $\rho$ is the ratio between the maximum and the minimum degree of the nodes in $V$ and $\tk=t-k+1$.
\end{theorem}
\begin{proof}{}
By Lemma \ref{lem:lem1} we can restrict our attention to ($k,t$)-simple-monotone configurations for $G$. Therefore,   the set $V$ can be partitioned into $t+1$ subsets $\X_{-\tk},\X_{-\tk+1}, \ldots, \X_{k-1}$ where $\tk=t-k+1$ and for $i=-\tk,-\tk+1,\ldots,k-1$, $\X_i$  denotes the set of nodes whose weight at round $j$ is  $\max(0,\min(j+i,k-1))$. Henceforth, we denote the size of $\X_i$ by $\x_i$ and  the sum of the degree of nodes in  $A\subseteq V$ by $d(A)$.
\\
In order to prove the theorem, we first show that, for each $i=-\tk,-\tk+1,\ldots,k-2$, it holds
\begin{equation}\label{x2x}\x_i \leq 2\rho\x_{k-1}.\end{equation}

Let  $E(A,B) =|\{e=(u,v)\in E \  \colon \ u\in A  \mbox{ and }  v \in B\}|$ denote the number of edges between a node  in $A$ and one in $B$.
Each node  $v\in \X_i$ must increase its weight for each round $r$ such that $0<r+i<k-1$; hence,   at round $r=\max(-i+1,0)$, node $v$ must have at least $\lceil  d(v)/2 \rceil$ neighbors which belong to $\bigcup_{j=i+1}^{k-1} \X_{j}$. Overall the number of edges between $\X_i$ and $ \bigcup_{j=i+1}^{k-1} \X_{j}$ satisfies
\begin{equation}\label{e1}
 E\left(\X_i, \bigcup_{j=i+1}^{k-1} \X_{j}\right)\geq   \frac{d(X_i)}{2}\geq  \frac{|X_i|d_{min}}{2}  =\frac{x_id_{min}}{2},
 \end{equation}
where $d_{min}$ represents the minimum degree of a node in $G$. Moreover, for each $i=-\tk,-\tk+1,\ldots,k-2$, the number of edges between $\X_i$ and $ \bigcup_{j=i+1}^{k-1} \X_{j}$ is \vspace*{-0.1truecm}
\begin{eqnarray*}\label{e2}
E\left(\X_i, \bigcup_{j=i+1}^{k-1} \X_{j}\right) &\leq&   \sum_{j=i+1}^{k-1} d(\X_j) - 2  E\left(\X_{i+1},\bigcup_{j=i+1}^{k-1} \X_{j}\right)
                 -2 E\left(\X_{i+2}, \bigcup_{j=i+2}^{k-1} \X_{j}\right) - \ldots \\
& &\qquad\qquad
    \ldots  -2 E\left(\X_{k-2}, \X_{k-2} \cup \X_{k-1}\right)  -2 E\left(\X_{k-1}, \X_{k-1}\right) \\
&\leq& \sum_{j=i+1}^{k-1} d(\X_j) - 2 \left[ E\left(\X_{i+1},\bigcup_{j=i+2}^{k-1} \X_{j}\right) + E\left(\X_{i+2}, \bigcup_{j=i+3}^{k-1} \X_{j}\right) +
\right. \\
& &\qquad\qquad \left.\ldots + E(\X_{k-2}, \X_{k-1}) \right ] \\
&\leq&   \sum_{j=i+1}^{k-1} d(\X_j) -2\left[ d(\X_{i+1})/2  + d(\X_{i+2})/2  + \ldots +  d(\X_{k-2})/2   \right ]  \\
&=& d(\X_{k-1}) \leq d_{max}|X_{k-1}|=d_{max}x_{k-1},
\end{eqnarray*}
where $d_{max}$ is the maximum node degree of a node in $G$. By this and  (\ref{e1}), recalling that $\rho=d_{max}/d_{min}$, we get  (\ref{x2x}).

\medskip
  Define  now $\y_i=\x_i/\x_{k-1}$. 
  By (\ref{x2x}),  $0 \leq \y_i \leq 2\rho$.
Our goal is to minimize the weight function
$ w(\C)=\sum_{j=1}^{k-1}j\x_j = \x_{k-1} \left((k-1)+\sum_{j=1}^{k-2} j\y_j   \right) $ with $|V|=\sum_{j=-\tk}^{k-1}\x_j = \x_{k-1} \left(1+\sum_{j=-\tk}^{k-2}\y_j  \right).$
Hence, $\x_{k-1}=\frac{|V|}{1+\sum_{j=-\tk}^{k-2}\y_j }$ and we can write
\begin{equation} \label{wf}
 w(C)= |V| \times \frac{ k-1+\sum_{j=1}^{k-2} j\y_j}{1+\sum_{j=-\tk}^{k-2}\y_j }.
\end{equation}

We distinguish now two cases depending on whether $t \geq k-1$ or $t<k-1$.

\medskip
\noindent
\textbf{Case I ($t\geq k-1$):}
In this case, it is possible to show that the rightmost term of (\ref{wf}) is minimized when
\begin{equation}
y_i=\begin{cases}{2\rho} & {\hbox{ if } -s\leq i\leq \opt} \cr
 {0} & {\hbox{ if } \opt < i\leq k-2,}\end{cases}\end{equation}
where $\opt{=}\left\lfloor\frac{\sqrt{(2\rho\tk+ \rho +1)^2+4\rho(k-1)}-(2\rho\tk+ \rho +1)}{2\rho}\right\rfloor$ is the floor of the positive root of the  equation $\rho i^2 +(2\rho\tk{+} \rho {+}1)i -(k{+}1).$\\
Let  $f(\y_{-\tk},\y_{-\tk+1}, \ldots, \y_{k-2})=\frac{k-1+\sum_{j=1}^{k-2} j\y_j}{1+\sum_{j=-\tk}^{k-2}\y_j }.$ This function is decreasing in $y_i$ for each $-\tk \leq i \leq 0$.
Hence, since $0 \leq \y_j \leq 2\rho$ for each $j$,\\
\centerline{$f(\y_{-\tk},\y_{-\tk+1}, \ldots,\y_0,\y_1, \ldots, \y_{k-2}) \geq  f(2\rho,2\rho,\ldots,2\rho,\y_1, \ldots, \y_{k-2}).$}
\smallskip \noindent
Moreover, we  show that the following two inequalities hold: \vspace*{-0.3truecm}
\begin{eqnarray}  \label{eq1}
\nonumber	f(2\rho,2\rho, \ldots,2\rho,\y_1, \ldots, \y_{\opt},\ldots, \y_{k-2}) &\geq&  f(2\rho,2\rho, \ldots,2\rho,\y_2, \ldots, \y_{\opt},\ldots, \y_{k-2}) \\
\nonumber	 &\geq&  f(2\rho,2\rho, \ldots,2\rho,\y_3, \ldots, \y_{\opt},\ldots, \y_{k-2}) \\
\nonumber &\geq& ...\\
	&\geq& f(2\rho,2\rho, \ldots, 2\rho, \y_{\opt+1},\ldots \y_{k-2})
\end{eqnarray}
\begin{eqnarray} \label{eq2}\vspace*{-0.4truecm}
\nonumber f(2\rho,2\rho, \ldots, 2\rho, \y_{\opt+1},\ldots \y_{k-2}) &\geq& f(2\rho,2\rho, \ldots, 2\rho, \y_{\opt+1},\ldots \y_{k-3},0)\\
\nonumber &\geq& f(2\rho,2\rho, \ldots, 2\rho, \y_{\opt+1},\ldots \y_{k-4},0,0) \\
\nonumber &\geq& \ldots\\
&\geq&  f(2\rho,2\rho, \ldots, 2\rho, 0,0,\ldots,0).\vspace*{-0.2truecm}
\end{eqnarray}
%
We first prove (\ref{eq1}). Each inequality  in (\ref{eq1}) is obtained by considering  the following one for some  $i\leq \opt$ (recalling that $\opt$ is the floor of the positive root of the  equation $\rho i^2 +(2\rho\tk{+} \rho {+}1)i -(k{+}1)$) \vspace*{-0.2truecm}
\begin{eqnarray} \label{eq3}
	f(2\rho, \ldots, 2\rho, \y_{i},\ldots, \y_{k-2}){=} \frac{A{+}i\y_i}{B{+}\y_i} {\geq} \frac{A{+}2\rho i}{B{+}2\rho}{=}f(2\rho, \ldots, 2\rho, \y_{i+1},\ldots, \y_{k-2})\vspace*{-0.2truecm}
\end{eqnarray}
where $A=k{-}1+\sum_{j=i+1}^{k-2} j\y_j + \rho i(i-1)$ and $B=1+\sum_{j=i+1}^{k-2}\y_j +2\rho(i+\tk).$

\noindent
\\
We notice that (\ref{eq3}) is satisfied whenever $\y_i(A-iB)\leq 2\rho(A-iB) $ and that for $i\leq \opt$
\vspace*{-0.2truecm}
\begin{eqnarray*}
A-iB &=&  k-1+\sum_{j=i+1}^{k-2} j\y_j + \rho i(i-1) - i\left( 1+\sum_{j=i+1}^{k-2}\y_j +2\rho(i+\tk) \right)\\
     &=& k-1+  \sum_{j=i+1}^{k-2} (j-i) \y_j	+ \rho i^2-\rho i-i-2\rho i^2-2\rho i\tk \\	
     &\geq&  -\rho i^2 -(2\rho\tk+ \rho +1)i +k-1 \geq 0.
 \end{eqnarray*}
Hence, 	(\ref{eq3}) and consequently  (\ref{eq1}) are satisfied. In order to get  (\ref{eq2}),  we show that for each $i > \opt$
\vspace*{-0.2truecm}
\begin{eqnarray} \label{eq4}
\hspace*{-0.5truecm}\nonumber	f(2\rho, \ldots, 2\rho, \y_{\opt+1},\ldots, \y_{i}, 0,\ldots,0) &{=}& \frac{C{+}i\y_i}{D{+}\y_i} 	 \\ &{\geq}& \frac{C}{D}{=} f(2\rho, \ldots, 2\rho, \y_{\opt+1},\ldots, \y_{i-1},0, \ldots,0)
\end{eqnarray}
where $C=k{-}1+\sum_{j=\opt+1}^{i-1} j\y_j + \rho \opt(\opt+1)$ and $D=1+\sum_{j=\opt+1}^{i-1}\y_j +2\rho(\tk{+}\opt{+}1)$.

\noindent
Since  (\ref{eq4}) is satisfied whenever $ \y_i (C-iD)\leq 0$ and since  now  $i> \opt$ we get
\begin{eqnarray*}\vspace*{-0.2truecm}
  C-iD   &=&  k-1+\sum_{j=\opt+1}^{i-1}j\y_j+\rho \opt(\opt+1)-i\left(1+\sum_{j=\opt+1}^{i-1}\y_j+2 \rho (\tk+\opt+1)\right)\\
  &\leq&  k-1+ \rho \opt^2+ \rho \opt - (\opt+1)-2 \rho(\opt+1)\tk-2 \rho(\opt+1)\opt- 2 \rho(\opt+1)\\
   &=& -\rho \opt^2 -(2\rho\tk+3\rho+1)\opt +k -2\rho \tk -2 \rho -2 \leq 0.
\end{eqnarray*}
 Hence,	(\ref{eq4}) and consequently  (\ref{eq2}) are satisfied.
Summarizing, we have that the minimizing values are
$$ \x_{i} = \left\{  \begin{array}{l l}
    \frac{|V|}{1+\sum_{j=-s}^{k-2}\y_j }=\frac{|V|}{2 \rho(\opt+\tk+1)+1}, & \quad \text{for $i=k-1$} \\
   2 \rho  x_{k-1}=\frac{2\rho|V|}{2 \rho(\opt+\tk+1)+1} & \quad \text{for $i=-s,-\tk+1,\ldots,\opt$} \\
    0 & \quad \text{otherwise.}
  \end{array} \right.$$
Therefore,
\begin{eqnarray*}
\sum_{j=1}^{k-1}j\x_j
 &=&\frac{|V|}{2 \rho(\opt{+}\tk{+}1){+}1}\left(k{-}1{+}2 \rho \sum_{j=1}^{\opt}j\right)
 =\frac{|V|}{2 \rho(\opt{+}\tk{+}1){+}1}\times \left( k{-}1{+}\rho\opt(\opt{+}1)  \right),
\end{eqnarray*}
and we can conclude that $w(\C)\geq \frac{|V|}{2 \rho(\opt+\tk+1)+1}\times \left( k-1+\rho \opt(\opt+1)  \right)$, when $t\geq k-1$. 

\medskip
\noindent
\textbf{Case II ($t< k-1$):} The proof of this case is left to the reader. 
\end{proof}

\begin{corol} \label{cor1}
Consider an undirected connected $d$-regular graph  $G=(V,E)$. Let $k\geq2$ and $t\geq1$ be integers.
Any ($k,t$)-dynamo $\C$, with $\lambda=1/2$, has weight
$$ w(\C)\geq  \left\{  \begin{array}{l  l}
     \frac{|V|}{2\opt+2\tk+3}\times \left(k{-}1{+}\opt(\opt{+}1)  \right) \text{ where } \opt=\lfloor\sqrt{t{+}1{+}\tk^2{+}\tk}\rfloor{-}(\tk{+}1) & \text{if  $t\geq k-1$} \\ \\
     \frac{|V|}{2\opt+2\tk+3}\times \left(k{-}1{+}\opt(\opt{+}1)-\tk(\tk{+}1)  \right) \text{ where } \opt=\lfloor\sqrt{t{+}1}\rfloor{-}(\tk{+}1) &\text{otherwise,}
  \end{array} \right.$$
where $\tk=t-k+1.$
\end{corol}
\vspace*{-0.2truecm}
We are now able to answer the question: {\em Which is the smallest value of $t$ such that the optimal dynamo contains only two weights?}
By analyzing the value of $\opt$ in the case $t\geq k-1$ we have that whenever  $t > \frac{k (2 \rho +1)- 2\rho -4}{2 \rho}$ then $\opt=0$, hence only the weights $0$ and $k-1$ will appear in the optimal configuration. When $\rho=1$ (i.e., on regular graphs) one has $t > \frac{3}{2}k-3$.

\vspace*{-0.2truecm}
\begin{remark}
Our result generalizes the one in {\em \cite{Santoro03}} with $k=2$. Indeed, when  $t\geq k-1=1$  by the above consideration we get $t>\frac{3}{2}k-3=0$ and $\opt=0$. Hence,   $w(\C)\geq \frac{|V|}{2s+3} \times(k-1)=\frac{|V|}{2t+1}$.
\end{remark}

\begin{theorem} \label{thm:t_large}
Let $G=(V,E)$ be an undirected connected graph, if $t$ is sufficiently large, then:\vspace*{-0.2truecm}
\begin{description}	
\item [(i)] any optimal $(k,t)$-dynamo   contains only the weights $0$ and $k-1$; 	
\item [(ii)] let $k \geq 2$ be an integer  and  $\C_2$  a 2--dynamo on $G$. Let $\C_k$ be  obtained from $\C_2$ by replacing the weight $1$ with the weight  $k-1$.  If $\C_2$ is an optimal $2$-dynamo then $\C_k$ is an optimal $k$-dynamo. Moreover, $w(\C_k)=w(\C_2)\times (k-1)$ and  $t(\C_k) =t(\C_2)+k-2$ (where $t(\C)$ is the time needed to reach the final configuration). 	\vspace*{-0.2truecm}
\end{description}\noindent Proof omitted.
\end{theorem}
\vspace*{-0.6truecm}
\section{Building ($k,t$)-dynamo}\label{sec4}\vspace*{-0.2truecm}
In this section we provide several optimal (or almost  optimal) ($k,t$)-dynamo constructions for Rings and Tori ($\lambda =1/2$ ) and Cliques (any $\lambda$).

\subsection{Rings} \vspace*{-0.2truecm}
A $n$-node ring $\R{n}$ consists of $n$ nodes and $n-1$ edges, where for $i=0,1,\ldots,n-1$ each node $v_i$ is connected with $v_{(i-1)\bmod n}$ and $v_{(i+1)\bmod n}$.

A necessary condition for $\C(\R{n},k)$ to be a $k$-dynamo ($\lambda\leq 1/2$) is that at least one node of $\R{n}$ is weighted by $k-1$. This condition is also sufficient.

\begin{figure}[th!]
	\centering
		\includegraphics[width=0.60\linewidth]{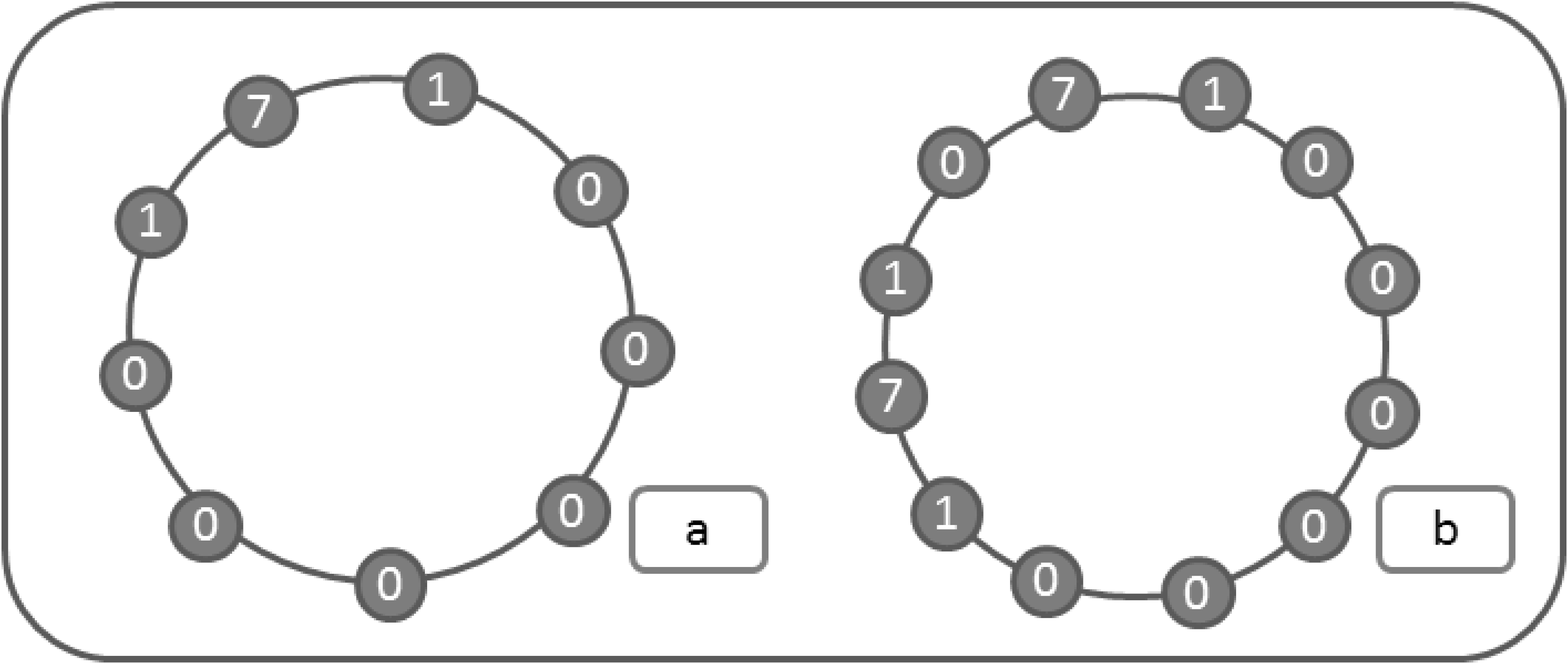}
	\hspace{0.06truecm}
		\includegraphics[width=0.60\linewidth]{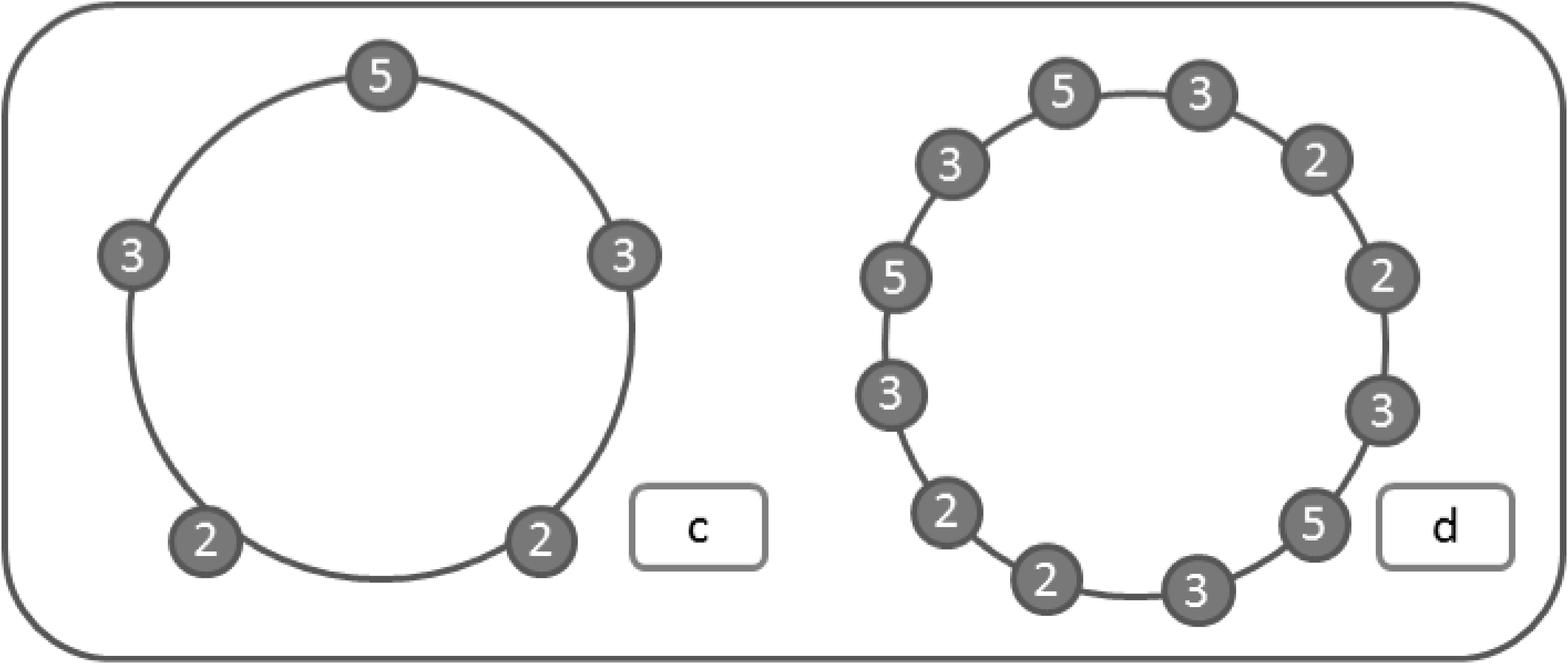}
		\vspace*{-0.6truecm}
		\caption{ {\bf ($k,t$)-dynamos  on Rings}: (a) $\C(\R{9},8,9)$, a ($8$,$9$)-dynamo  on $\R{9}$ ($\opt=1$), in this particular case $n=2\opt+2\tk+3$; (b) $\C(\R{12},8,9)$ a ($8$,$9$)-dynamo  on $\R{12}$ ($\opt=1$);
		(c) $ \C(\R{5},6,3)$, a ($6$,$3$)-dynamo  on $\R{5}$ ($\opt=3$), in this particular case $n=2\opt+2\tk+3$; (d) $ \C(\R{12},6,3)$,  a ($6$,$3$)-dynamo  on $\R{12}$  ($\opt=3$).\label{rings} \vspace*{-0.5truecm}	}
\end{figure}

\begin{theorem}
An optimal $k$-dynamo ($\lambda\leq 1/2$) $\C(\R{n},k)$  has weight $w(\C(\R{n},k))=(k-1)$, and it reaches its final configuration within  $t=k-2+\lceil{\frac{n-1}{2}}\rceil$ rounds.
\end{theorem}

A ($k,t$)-dynamo ($\lambda=1/2$)  for a ring $\R{n}$ is obtained by the following  partition of $V$ which defines the initial configuration (see Figure \ref{rings}) $\C(\R{n},k,t)$: for $i=0,1,\ldots,n$,
$$ \forall v_i\in \R{n}, \ \ v_i \in \left\{  \begin{array}{l l}
    \X_{k-1} & \quad \text{if $j=0$}\\
   \X_{\opt+1-j} & \quad \text{if $1\leq j \leq \opt+\tk+1$}\\
    \X_{j-\opt-2\tk-2} & \quad \text{if $\opt +\tk+2\leq j \leq 2\opt+2\tk+2$}
  \end{array} \right.$$
where $\tk=t-k+1,$ $j=i\bmod(2\ell+2\tk+3)$ and $\opt=\lfloor\sqrt{t+1+\tk^2+\tk}\rfloor-(\tk+1)$ if $t\geq k-1$ and $\opt=\lfloor\sqrt{t+1}\rfloor-(\tk+1)$ otherwise.

\begin{theorem} \label{th:ring_irr}
(i) The configuration $\C(\R{n},k,t)$ is a ($k,t$)-dynamo for any value of $n$, $\lambda=1/2$, $k\geq2$ and $t\geq 1$.
(ii) The weight of  $\C(\R{n},k,t)$ is
$$ w(\C(\R{n},k,t))\leq  \left\{  \begin{array}{l  l}
    \left \lceil \frac{n}{2\opt+2\tk+3} \right \rceil \left( k{-}1+\opt(\opt{+}1)  \right) &\text{if  $t\geq k-1$} \\   \text{ \ \ \ \ where } \opt=\lfloor\sqrt{t+1+\tk^2+\tk}\rfloor-(\tk{+}1)\\
    \left \lceil \frac{n}{2\opt+2\tk+3}\right \rceil \left( k{-}1+\opt(\opt{+}1)-\tk(\tk{+}1)  \right) &\text{otherwise} \\ \text{ \ \ \ \ where } \opt=\lfloor\sqrt{t+1}\rfloor-(\tk{+}1)
  \end{array} \right.$$
\end{theorem}

\begin{proof}{}
(i) By construction $\C(\R{n},k,t)$  is ($k,t$)-simple-monotone, hence by
 Lemma \ref{lem:ub}, $\C(\R{n},k,t)$ is a  ($k,t$)-dynamo.
(ii) There are two cases to consider: if $t\geq k-1$, then starting from $v_0$ each set of $2\opt+2\tk+3$ nodes weights $k-1+2\sum_{i=1}^{\opt}i=k-1+\opt(\opt+1)$.
Then the weight of  $\C(\R{n},k,t)$ is smaller than the weight of $\C(\R{\overline{n}},k,t)$ where $\overline{n}= \lceil\frac{n}{2\opt+2\tk+3}\rceil \times (2\opt+2\tk+3).$
Hence, $w(\C(\R{n},k,t))\leq w(\C(\R{\overline{n}},k,t))  =\left \lceil \frac{n}{2\opt+2\tk+3} \right \rceil \left( k{-}1+\opt(\opt{+}1)  \right)$.
Similarly for $t< k-1.$ 
\end{proof}
By Corollary \ref{cor1} and Theorem \ref{th:ring_irr} we have the following Corollary.
\begin{corol}
When $n/(2\opt+2\tk+3)$ is integer,  $\C(\R{n},k,t)$ is an optimal ($k,t$)-dynamo.
\end{corol}

\subsection{Tori} \vspace*{-0.2truecm}
A $n \times m$-node tori $\T{n}{m}$ consists of $n\times m$ nodes and $2(n\times m)$ edges, where for $i=0,1,\ldots,n-1$ and $j=0,1,\ldots,m-1$, each node $v_{i,j}$ is connected with four nodes: $v_{i,(j-1)\bmod m}$, $v_{i,(j+1)\bmod m}$, $v_{(i-1)\bmod n,j}$ and $v_{(i+1)\bmod n,j}$.
\begin{figure}[th!]
	\centering
\includegraphics[width=0.90\linewidth]{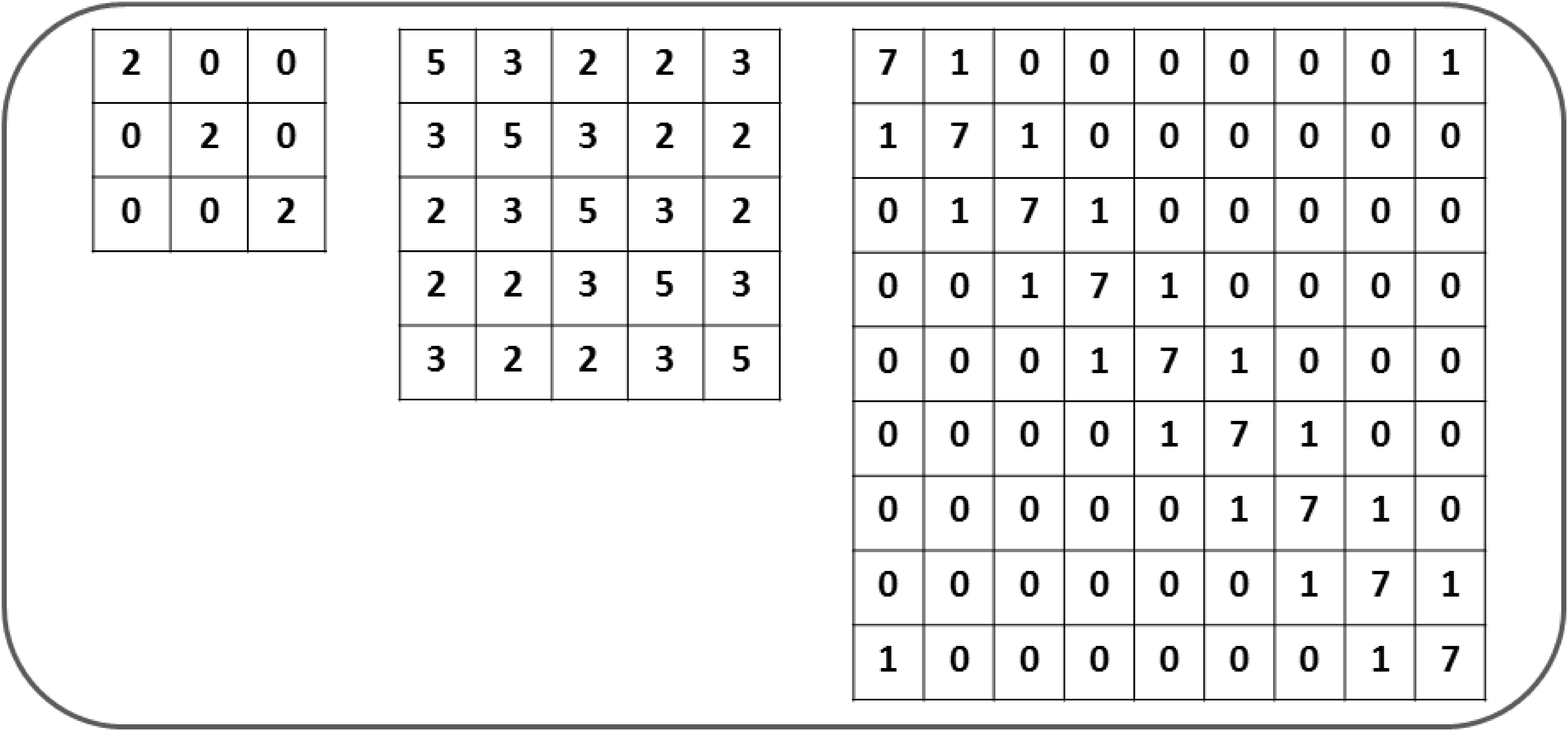}
	\vspace*{-0.3truecm}
		\caption{{\bf ($k,t$)-dynamos  on Tori}: (left) $\C(\T{3}{3},3,2)$, a ($3$,$2$)-dynamo  on $\T{3}{3}$ ($\opt{=}0$);
		(middle) $\C(\T{5}{5},6,3)$, a ($6$,$3$)-dynamo   on $\T{5}{5}$ ($\opt{=}3$);
			(right)	$\C(\T{9}{9},8,9)$ a ($8$,$9$)-dynamo   on $\T{9}{9}$ ($\opt{=}1$).
			\label{tori} \vspace*{-0.4truecm}}
\end{figure}

A ($k,t$)-dynamo ($\lambda=1/2$) for $\T{2\opt+2\tk+3}{2\opt+2\tk+3}$ is obtained by weighting diagonals with the same order defined for dynamos on rings. Specifically, the configuration \\ $\C(\T{2\opt+2\tk+3}{2\opt+2\tk+3},k,t)$ is defined by the partition of $V$ described  as follows, let $D_i=\{v_{a,b}\ : \ i=(b-a) \bmod (2\opt+2\tk +3)  \}$ denote the $i$-th diagonal of $\T{2\opt+2\tk+3}{2\opt+2\tk+3}$, for $i=0,1,\ldots ,2\opt+2\tk+2,$
 $$\forall v \in D_i, \ \ v  \in \left\{  \begin{array}{l l}
    \X_{k-1} & \quad \text{if $i=0$}\\
    \X_{\opt+1-i} & \quad \text{if $1\leq i \leq \opt+\tk+ 1$}\\
     \X_{i-\opt-2\tk-2} & \quad \text{if $\opt+\tk +2\leq i \leq 2\opt+2\tk+2,$}
\end{array} \right.$$
where $\tk=t-k+1,$  $\opt=\lfloor\sqrt{t+1+\tk^2+\tk}\rfloor-(\tk+1)$ if $t\geq k-1$ and $\opt=\lfloor\sqrt{t+1}\rfloor-(\tk+1)$ otherwise.
Some examples are depicted in Figure \ref{tori}.

\begin{theorem}
The configuration $\C(\T{2\opt+2\tk+3}{2\opt+2\tk+3},k,t)$ is an optimal ($k,t$)-dynamo for any $k\geq2$, $t\geq1$ and $\lambda=1/2$.
\end{theorem}
\begin{proof}{}
Let $\C=\C(\T{2\opt+2\tk+3}{2\opt+2\tk+3},k,t)$. By construction $\C$ is ($k,t$)-simple-monotone, hence by
 Lemma \ref{lem:ub}, it is a ($k,t$)-dynamo.
To show its optimality we  distinguish two cases.
If $t\geq k-1,$ each row (resp. each column) corresponds to $\C(\R{2\opt+2\tk+3},k,t)$ and  its weight is  $k-1+\opt(\opt+1)$. Overall, $w(\C)= (2\opt+2\tk+3) \times \left( k-1+\opt(\opt+1)\right)$ that matches the bound  in Corollary \ref{cor1}.
Similarly for $t< k-1$. 
\end{proof}


A ($k,t$)-dynamo  for $\T{n}{m}$ is obtained by building a grid $\lceil \frac{n}{2\opt+2\tk+3} \rceil \times \lceil \frac{m}{2\opt+2\tk+3} \rceil$,  where each cell is filled with a configuration $\C(\T{2\opt+2\tk+3}{2\opt+2\tk+3},k,t)$ defined above. Then, the exceeding part is removed and the last row and the last column are updated. In particular,  for each column (resp. row), if the removed part contains a $k-1$, then the element in the last row  (resp. column) is given  the value $k-1$ (see Figure \ref{tori3}).
We call this configuration $\C(\T{n}{m},k,t)$.

\begin{figure}[th!]
	\centering
		\includegraphics[width=0.90\linewidth]{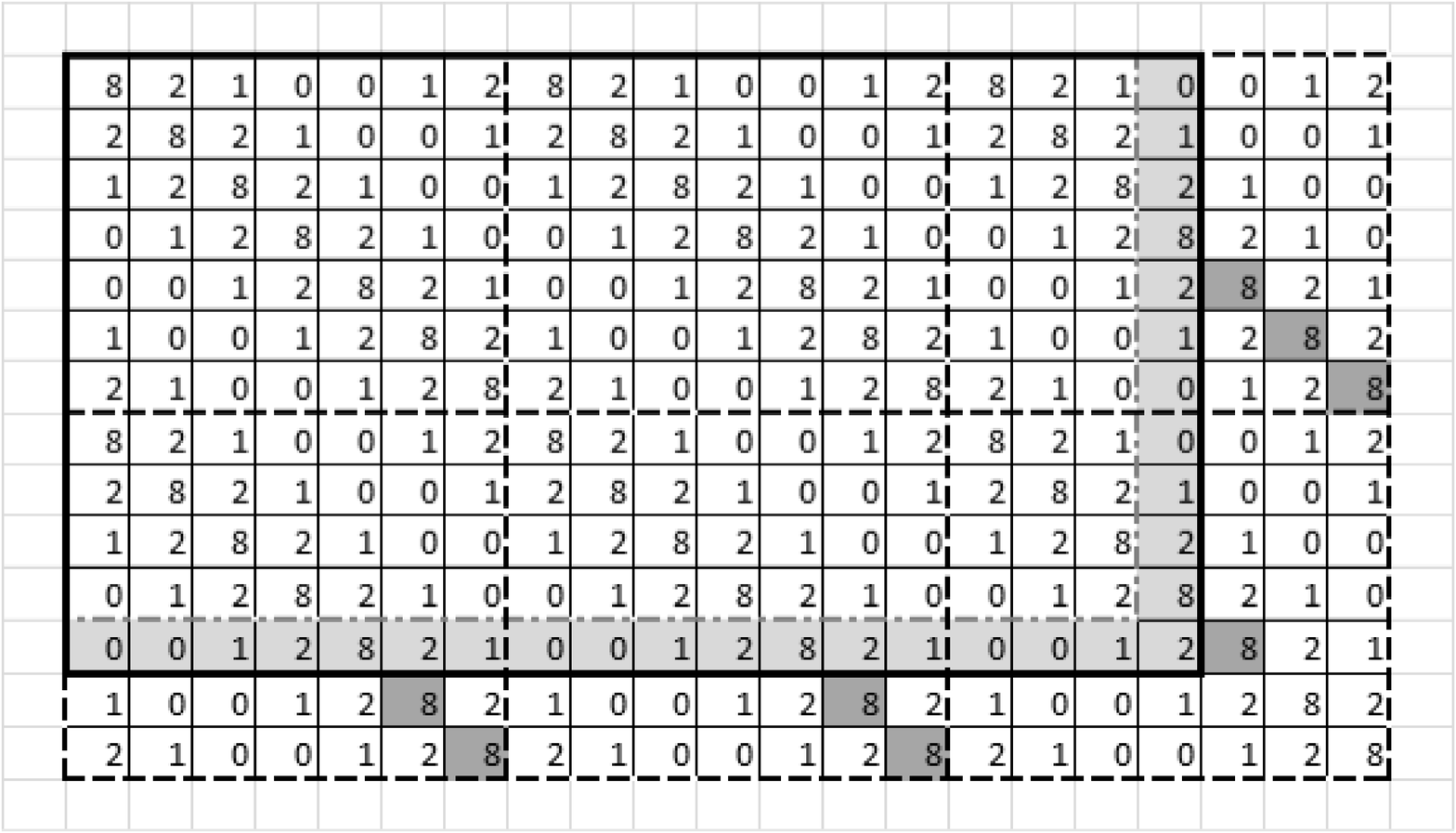}	
		\includegraphics[width=0.90\linewidth]{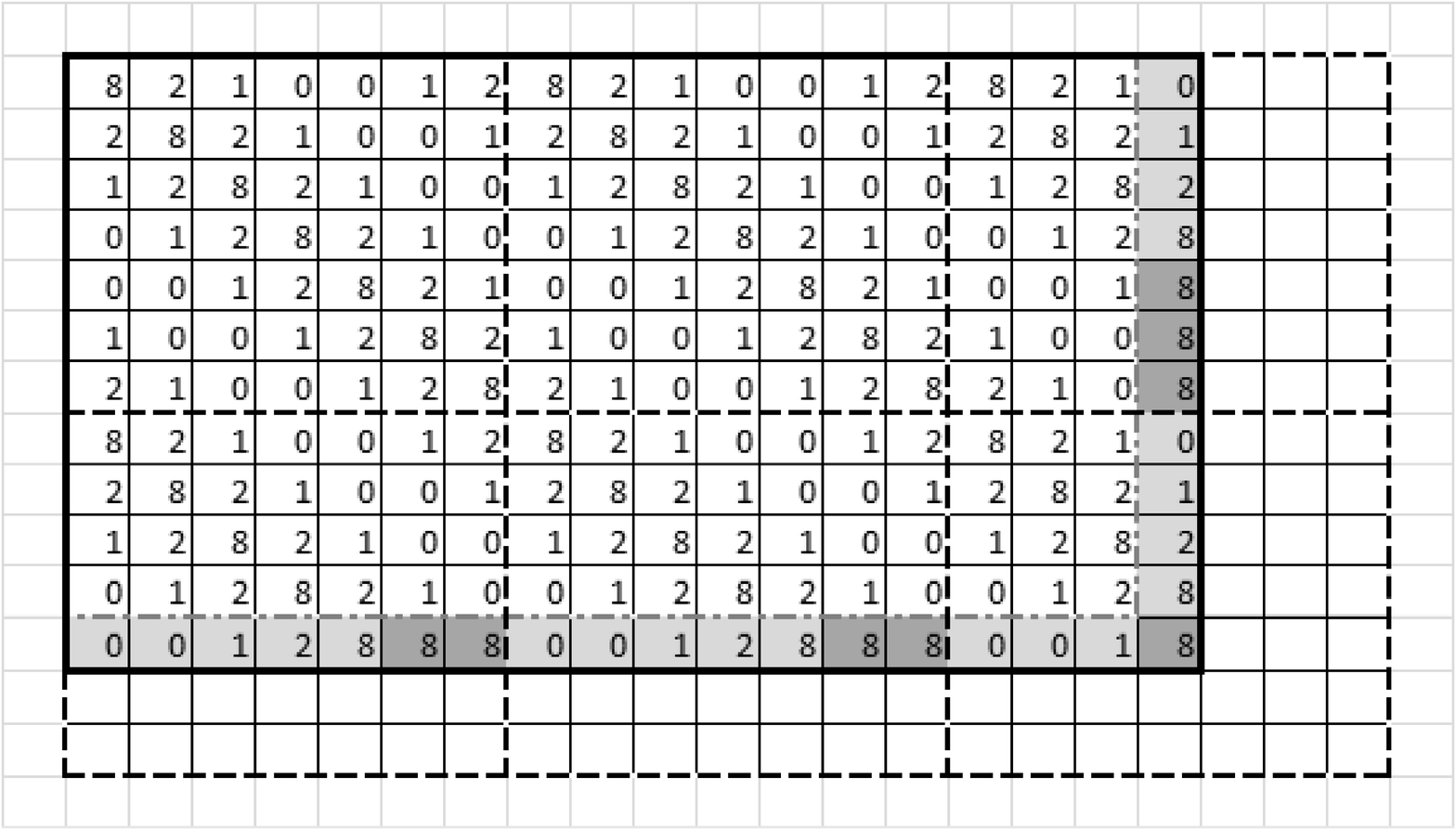}
		\vspace*{-0.3truecm}
		\caption{ $\C(\T{12}{18},9,8)$,  a ($9,8$)-dynamo  on $\T{12}{18}$ ($\opt=2$): (left) a grid $2 \times 3$ is filled with 6 configuration $\C(\T{7}{7},9,8);$ (right) The exceeding parts i.e., the last two rows and the last three columns are removed. Finally the last row and the last column are updated in order to obtain a configuration that satisfies Lemma \ref{lem:ub}.  \label{tori3} \vspace*{-0.4truecm}	}
\end{figure}

\begin{theorem} \label{th:tori_irr}
$\ $ \par\noindent
(i)  $\C(\T{n}{m},k,t)$  is a ($k,t$)-dynamo  for any value of $n$, $m$,$\lambda=1/2$, $k\geq2$ and $t\geq1$.\\
(ii) The weight of  $\C(\T{n}{m},k,t)$  is
$$ w(\C(\T{n}{m},k,t))\leq  \left\{  \begin{array}{l  l}
    \lceil\frac{n}{2\opt{+}2\tk{+}3}\rceil  \lceil\frac{m}{2\opt{+}2\tk{+}3}\rceil  (2\opt{+}2\tk{+}3) \left( k{-}1{+}\opt(\opt{+}1)  \right)& \text{if  $t\geq k-1$}\\ \ \ \ \text{ where } \opt=\lfloor\sqrt{t{+}1{+}\tk^2{+}\tk}\rfloor{-}(\tk{+}1)  \\
    \lceil\frac{n}{2\opt{+}2\tk{+}3}\rceil  \lceil\frac{m}{2\opt{+}2\tk{+}3}\rceil  (2\opt{+}2\tk{+}3) \left( k{-}1{+}\opt(\opt{+}1){-}\tk(\tk{+}1)  \right) & \text{otherwise}\\ \ \ \ \text{ where } \opt=\lfloor\sqrt{t{+}1}\rfloor{-}(\tk{+}1).
  \end{array} \right.$$
\end{theorem}

\begin{proof}{}
(i) By construction $\C(\T{n}{m},k,t)$ is ($k,t$)-simple-monotone (cfr. Figure \ref{tori3}), hence by
 Lemma \ref{lem:ub}, $\C(\T{n}{m},k,t)$  is a ($k,t$)-dynamo.\\
(ii) The grid contains $\lceil\frac{n}{2\opt+2\tk+3}\rceil \times \lceil\frac{m}{2\opt+2\tk+3}\rceil$ cells.  If $t\geq k-1$,
each cell has weight\\ $w(\C(\T{2\opt+2\tk+3}{2\opt+2\tk+3},k,t))= (2\opt+2\tk+3) \times \left( k-1+\opt(\opt+1)\right).$
 \\
 Moreover, the nodes that change their weight take the weight of a removed element. Hence, the weight of $\C(\T{n}{m},k,t)$ is upper bounded by the weight of the full grid which is $\lceil\frac{n}{2\opt+2\tk+3}\rceil \times \lceil\frac{m}{2\opt+2\tk+3}\rceil \times w(\C(\T{2\opt+2\tk+3}{2\opt+2\tk+3},k,t))$.  Similarly for $t< k-1$. 
\end{proof}
By Corollary \ref{cor1} and Theorem \ref{th:tori_irr} we have the following Corollary.
\begin{corol}
If both $n$ and $m$ are  multiples of  $2\opt+2\tk+3$,  $\C(\T{n}{m},k,t)$ is an optimal ($k,t$)-dynamo.
\end{corol}

\subsection{Cliques} \vspace*{-0.2truecm}
Let $K_n$ be the clique on $n$ nodes. A necessary condition for a $k$-dynamo $\C(K_n,k)$ is that $\lceil \lambda(n-1) \rceil$ nodes  are weighted by $k-1$. The condition is also sufficient and if the remaining $\lfloor \lambda(n-1) \rfloor$ nodes are weighted by $0$, the $k$-dynamo is optimal and reaches its final configuration within $t=k-1$ rounds. So,  when $t\geq k-1$ the optimal configuration is obtained by weighting $\lceil \lambda(n-1) \rceil$ nodes by $k-1$ and the remaining nodes by $0$. For $t<k-1$, an optimal ($k$,$t$)-dynamo is obtained by assigning weight $k-t-1$ to all the non-$k-1$ weighted nodes. Clearly this configuration is optimal, if we assign a weight smaller than $k-t-1$ to a node $v$, then $v$ can not reach the weight $k-1$ within $t$ rounds. Therefore:
\begin{theorem}
Let $K_n$ be the clique on  $n$ nodes. An optimal ($k$,$t$)-dynamo $\C(K_n,k,t)$ has weight \\$w(\C(K_n,k,t))=(k-1) \times \lceil \lambda(n-1) \rceil + \max(k-t-1,0) \times \lfloor  \lambda(n-1)  \rfloor$.
\end{theorem}

\section{Conclusion and Open Problems} \vspace*{-0.2truecm}
In this work we studied multivalued dynamos with respect to both weight and  time. 
We derived lower bounds on the weight of ($k,t$)-dynamo and provided constructive tight upper bounds for rings, tori and cliques.
Several dimensions of the problem remain unexplored, different updating rules with could
be addressed, as for instance the case of reversible rules.
Finally, the behavior of this protocol on different topologies such as small world \cite{WATTS99}, scale-free \cite{Barabasi99}, and time-varying networks \cite{CFQS2012J}.

\paragraph{\textbf{Acknowledgments.}}
We would like to thank Ugo Vaccaro for many stimulating discussions and the anonymous referees whose helpful comments allowed to significantly improve the presentation of their work.
\bibliographystyle{plain}
\baselineskip=0.44truecm
{\footnotesize
\bibliography{wg2012}
}

\end{document}